%% file: paper.tex
\documentclass[conference,onecolumn,draft,a4paper]{IEEEtran}
\ifCLASSINFOpdf
\else
\fi

\usepackage{tikz}
\usetikzlibrary{decorations,shapes,backgrounds,calc}
\tikzstyle{msg}=[->,black,>=latex]
\tikzstyle{rubber}=[|<->|]
\tikzstyle{announce}=[draw=blue,fill=blue,shape=diamond,right,minimum
  height=2mm,minimum width=1.6667mm,inner sep=0pt]
\tikzstyle{decide}=[draw=red,fill=red,shape=isosceles triangle,right,minimum
  height=2mm,minimum width=1.6667mm,inner sep=0pt,shape border rotate=90]
\tikzstyle{cast}=[draw=green!50!black,fill=green!50!black,shape=circle,left,minimum
  height=2mm,minimum width=1.6667mm,inner sep=0pt]

\usepackage{multirow}
\usepackage{graphicx}
\usepackage{epstopdf}
\usepackage{amssymb}
\usepackage{rounddiag}
\graphicspath{{../}}

\usepackage{technote}
\usepackage{homodel}
\usepackage{enumerate}

\usepackage{caption}

\newconstruct{\FOREACH}{\textbf{for each}}{\textbf{do}}{\ENDFOREACH}{}

\newcommand\With{\textbf{while}}
\newcommand\From{\textbf{from}}
\newcommand\Broadcast{\textbf{broadcast}}

\newident{noop}
\newconstruct{\UPON}{\textbf{upon}}{\textbf{do}}{\ENDUPON}{}

\newconstruct{\FUNCTION}{\textbf{Function}}{\textbf{:}}{\ENDFUNCTION}{}

 \renewcommand{\note}[2][default]{\relax}

\newcommand{\tr}[1]{}
\renewcommand{\tr}[1]{#1}

\begin{document}
%
\title{The latest gossip on BFT consensus\vspace{-0.7\baselineskip}}

\author{\IEEEauthorblockN{\large Ethan Buchman, Jae Kwon and Zarko Milosevic\\}
	\IEEEauthorblockN{\large Tendermint}\\
	\IEEEauthorblockN{September 24, 2018}
}

\maketitle
\vspace*{0.5em}

\begin{abstract}
The paper presents Tendermint, a new protocol for ordering events in a distributed network under adversarial conditions. More commonly known as Byzantine Fault Tolerant (BFT) consensus or atomic broadcast, the problem has attracted significant attention in recent years due to the widespread success of blockchain-based digital currencies, such as Bitcoin and Ethereum, which successfully solved the problem in a public setting without a central authority. Tendermint modernizes classic academic work on the subject and simplifies the design of the BFT algorithm by relying on a peer-to-peer gossip protocol among nodes. 
\end{abstract}


\input{intro}

\input{definitions}

\input{consensus}

\input{proof}

\input{conclusion}

\bibliographystyle{IEEEtran}
\bibliography{lit}


\end{document}

%% file: intro.tex
\section{Introduction} \label{sec:tendermint}

Consensus is one of the most fundamental problems in distributed computing. It
is important because of it's role in State Machine Replication (SMR), a generic
approach for replicating services that can be modeled as a deterministic state
machine~\cite{Lam78:cacm, Sch90:survey}. The key idea of this approach is that
service replicas start in the same initial state, and then execute requests
(also called transactions) in the same order; thereby guaranteeing that
replicas stay in sync with each other. The role of consensus in the SMR
approach is ensuring that all replicas receive transactions in the same order.
Traditionally, deployments of SMR based systems are in data-center settings
(local area network), have a small number of replicas (three to seven) and are
typically part of a single administration domain (e.g., Chubby
\cite{Bur:osdi06}); therefore they handle benign (crash) failures only, as more
general forms of failure (in particular, malicious or Byzantine faults) are
considered to occur with only negligible probability.  

The success of cryptocurrencies or blockchain systems in recent years (e.g.,
\cite{Nak2012:bitcoin, But2014:ethereum}) pose a whole new set of challenges on
the design and deployment of SMR based systems: reaching agreement over wide
area network, among large number of nodes (hundreds or thousands) that are not
part of the same administration domain, and where a subset of nodes can behave
maliciously (Byzantine faults). Furthermore, contrary to the previous
data-center deployments where nodes are fully connected to each other, in
blockchain systems, a node is only connected to a subset of other nodes, so
communication is achieved by gossip-based peer-to-peer protocols. 
The new requirements demand designs and algorithms that are not necessarily
present in the classical academic literature on Byzantine fault tolerant
consensus (or SMR) systems (e.g., \cite{DLS88:jacm, CL02:tcs}) as the primary 
focus was different setup. 

In this paper we describe a novel Byzantine-fault tolerant consensus algorithm
that is the core of the BFT SMR platform called Tendermint\footnote{The
	Tendermint platform is available open source at
	https://github.com/tendermint/tendermint.}. The Tendermint platform consists of
a high-performance BFT SMR implementation written in Go, a flexible interface
for
building arbitrary deterministic applications above the consensus, and a suite
of tools for deployment and management.  

The Tendermint consensus algorithm is inspired by the PBFT SMR
algorithm~\cite{CL99:osdi} and the DLS algorithm for authenticated faults (the
Algorithm 2 from \cite{DLS88:jacm}). Similar to DLS algorithm, Tendermint
proceeds in
rounds\footnote{Tendermint is not presented in the basic round model of
	\cite{DLS88:jacm}. Furthermore, we use the term round differently than in
	\cite{DLS88:jacm}; in Tendermint a round denotes a sequence of communication
	steps instead of a single communication step in \cite{DLS88:jacm}.}, where each
round has a dedicated proposer (also called coordinator or
leader) and a process proceeds to a new round as part of normal
processing (not only in case the proposer is faulty or suspected as being faulty
by enough processes as in PBFT).  
The communication pattern of each round is very similar to the "normal" case
of PBFT. Therefore, in preferable conditions (correct proposer, timely and
reliable communication between correct processes), Tendermint decides in three
communication steps (the same as PBFT). 

The major novelty and contribution of the Tendermint consensus algorithm is a
new termination mechanism. As explained in \cite{MHS09:opodis, RMS10:dsn}, the
existing BFT consensus (and SMR) algorithms for the partially synchronous
system model (for example PBFT~\cite{CL99:osdi}, \cite{DLS88:jacm},
\cite{MA06:tdsc}) typically relies on the communication pattern illustrated in
Figure~\ref{ch3:fig:coordinator-change} for termination. The
Figure~\ref{ch3:fig:coordinator-change} illustrates messages exchanged during
the proposer change when processes start a new round\footnote{There is no
	consistent terminology in the distributed computing terminology on naming
	sequence of communication steps that corresponds to a logical unit. It is
	sometimes called a round, phase or a view.}. It guarantees that eventually (ie.
after some Global Stabilization Time, GST), there exists a round with a correct
proposer that will bring the system into a univalent configuration.
Intuitively, in a round in which the proposed value is accepted
by all correct processes, and communication between correct processes is
timely and reliable, all correct processes decide.

\begin{figure}[tbh!]
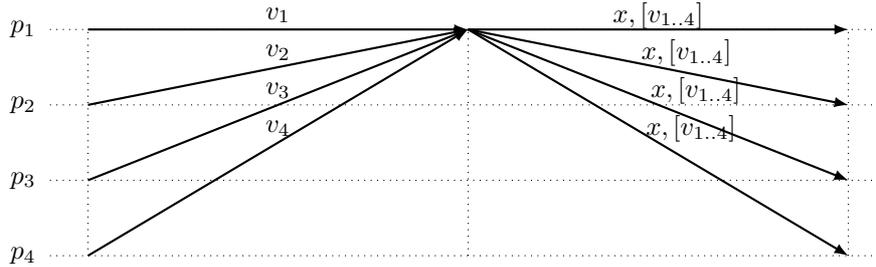
 \def\rdstretch{5} \def\ystretch{3} \centering
	\begin{rounddiag}{4}{2} \round{1}{~} \rdmessage{1}{1}{$v_1$}
		\rdmessage{2}{1}{$v_2$} \rdmessage{3}{1}{$v_3$} \rdmessage{4}{1}{$v_4$}
		\round{2}{~} \rdmessage{1}{1}{$x, [v_{1..4}]$}
		\rdmessage{1}{2}{$~~~~~~x, [v_{1..4}]$} \rdmessage{1}{3}{$~~~~~~~~x,
			[v_{1..4}]$} \rdmessage{1}{4}{$~~~~~~~x, [v_{1..4}]$} \end{rounddiag}
	\vspace{-5mm} \caption{\boldmath Proposer (coordinator) change: $p_1$ is the
		new proposer.} \label{ch3:fig:coordinator-change} \end{figure}  

To ensure that a proposed value is accepted by all correct
processes\footnote{The proposed value is not blindly accepted by correct
	processes in BFT algorithms. A correct process always verifies if the proposed
	value is safe to be accepted so that safety properties of consensus are not
	violated.}
a proposer will 1) build the global state by receiving messages from other
processes, 2) select the safe value to propose and 3) send the selected value
together with the signed messages
received in the first step to support it. The
value $v_i$ that a correct process sends to the next proposer normally
corresponds to a value the process considers as acceptable for a decision: 

\begin{itemize} \item in PBFT~\cite{CL99:osdi} and DLS~\cite{DLS88:jacm} it is
	not the value itself but a set of $2f+1$ signed messages with the same
	value id, \item in Fast Byzantine Paxos~\cite{MA06:tdsc} the value
	itself is being sent.  \end{itemize}

In both cases, using this mechanism in our system model (ie. high
number of nodes over gossip based network) would have high communication
complexity that increases with the number of processes: in the first case as
the message sent depends on the total number of processes, and in the second
case as the value (block of transactions) is sent by each process. The set of
messages received in the first step are normally piggybacked on the proposal
message (in the Figure~\ref{ch3:fig:coordinator-change} denoted with
$[v_{1..4}]$) to justify the choice of the selected value $x$. Note that
sending this message also does not scale with the number of processes in the
system.   

We designed a novel termination mechanism for Tendermint that better suits the
system model we consider. It does not require additional communication (neither
sending new messages nor piggybacking information on the existing messages) and
it is fully based on the communication pattern that is very similar to the
normal case in PBFT \cite{CL99:osdi}. Therefore, there is only a single mode of
execution in Tendermint, i.e., there is no separation between the normal and
the recovery mode, which is the case in other PBFT-like protocols (e.g.,
\cite{CL99:osdi}, \cite{Ver09:spinning} or \cite{Cle09:aardvark}). We believe
this makes Tendermint simpler to understand and implement correctly. 

Note that the orthogonal approach for reducing message complexity in order to
improve
scalability and decentralization (number of processes) of BFT consensus
algorithms is using advanced cryptography (for example Boneh-Lynn-Shacham (BLS)
signatures \cite{BLS2001:crypto}) as done for example in SBFT
\cite{Gue2018:sbft}.  

The remainder of the paper is as follows: Section~\ref{sec:definitions} defines
the system model and gives the problem definitions. Tendermint
consensus algorithm is presented in Section~\ref{sec:tendermint} and the
proofs are given in Section~\ref{sec:proof}. We conclude in
Section~\ref{sec:conclusion}.

%% file: definitions.tex
\section{Definitions} \label{sec:definitions}

\subsection{Model}

We consider a system of processes that communicate by exchanging messages.
Processes can be correct or faulty, where a faulty process can behave in an
arbitrary way, i.e., we consider Byzantine faults. We assume that each process
has some amount of voting power (voting power of a process can be $0$).
Processes in our model are not part of a single administrative domain;
therefore we cannot enforce a direct network connectivity between all
processes. Instead, we assume that each process is connected to a subset of
processes called peers, such that there is an indirect communication channel
between all correct processes. Communication between processes is established
using a gossip protocol \cite{Dem1987:gossip}.

Formally, we model the network communication using a variant of the \emph{partially
synchronous system model}~\cite{DLS88:jacm}: in all executions of the system
there is a bound $\Delta$ and an instant GST (Global Stabilization Time) such
that all communication among correct processes after GST is reliable and
$\Delta$-timely, i.e., if a correct process $p$ sends message $m$ at time $t
\ge GST$ to a correct process $q$, then $q$ will receive $m$ before $t +
\Delta$\footnote{Note that as we do not assume direct communication channels
    among all correct processes, this implies that before the message $m$
    reaches $q$, it might pass through a number of correct processes that will
forward the message $m$ using gossip protocol towards $q$.}. 
In addition to the standard \emph{partially
	synchronous system model}~\cite{DLS88:jacm}, we assume an auxiliary property 
that captures gossip-based nature of communication\footnote{The details of the Tendermint gossip protocol will be discussed in a separate
	technical report. }:

\begin{itemize} \item \emph{Gossip communication:} If a correct process $p$
	sends some message $m$ at time $t$, all correct processes will receive
	$m$ before $max\{t, GST\} + \Delta$. Furthermore, if a correct process $p$
	receives some message $m$ at time $t$, all correct processes will receive
	$m$ before $max\{t, GST\} + \Delta$.    \end{itemize}

The bound $\Delta$ and GST are system
parameters whose values are not required to be known for the safety of our
algorithm. Termination of the algorithm is guaranteed within a bounded duration
after GST.  In practice, the algorithm will work correctly in the slightly
weaker variant of the model where the system alternates between (long enough)
good periods (corresponds to the \emph{after} GST period where system is
reliable and $\Delta$-timely) and bad periods (corresponds to the period
\emph{before} GST during which the system is asynchronous and messages can be
lost), but consideration of the GST model simplifies the discussion.  

We assume that process steps (which might include sending and receiving
messages) take zero time.  Processes are equipped with clocks so they can
measure local timeouts.  
Spoofing/impersonation attacks are assumed to be impossible at all times due to
the use of public-key cryptography, i.e., we assume that all protocol messages contains a digital signature.
Therefore, when a correct
process $q$ receives a signed message $m$ from its peer, the process $q$ can
verify who was the original sender of the message $m$ and if the message signature is valid.
We do not explicitly state a signature verification step in the pseudo-code of the algorithm to improve readability;
we assume that only messages with the valid signature are considered at that level (and messages with invalid signatures
are dropped).



\subsection{State Machine Replication}

State machine replication (SMR) is a general approach for replicating services
modeled as a deterministic state machine~\cite{Lam78:cacm,Sch90:survey}.  The
key idea of this approach is to guarantee that all replicas start in the same
state and then apply requests from clients in the same order, thereby
guaranteeing that the replicas' states will not diverge.  Following
Schneider~\cite{Sch90:survey}, we note that the following is key for
implementing a replicated state machine tolerant to (Byzantine) faults:

\begin{itemize} \item \emph{Replica Coordination.} All [non-faulty] replicas
    receive and process the same sequence of requests.  \end{itemize}

Moreover, as Schneider also notes, this property can be decomposed into two
parts, \emph{Agreement} and \emph{Order}: Agreement requires all (non-faulty)
replicas to receive all requests, and Order requires that the order of received
requests is the same at all replicas.

There is an additional requirement that needs to be ensured by Byzantine
tolerant state machine replication: only requests (called transactions in the
Tendermint terminology) proposed by clients are executed. In Tendermint,
transaction verification is the responsibility of the service that is being
replicated; upon receiving a transaction from the client, the Tendermint
process will ask the service if the request is valid, and only valid requests
will be processed. 

 \subsection{Consensus} \label{sec:consensus}

Tendermint solves state machine replication by sequentially executing consensus
instances to agree on each block of transactions that are
then executed by the service being replicated. We consider a variant of the
Byzantine consensus problem called Validity Predicate-based Byzantine consensus
that is motivated by blockchain systems~\cite{GLR17:red-belly-bc}. The problem
is defined by an agreement, a termination, and a validity property.

 \begin{itemize} \item \emph{Agreement:} No two correct processes decide on
         different values.  \item \emph{Termination:} All correct processes
         eventually decide on a value.  \item \emph{Validity:} A decided value
             is valid, i.e., it satisfies the predefined predicate denoted
             \emph{valid()}.  \end{itemize}

 This variant of the Byzantine consensus problem has an application-specific
 \emph{valid()} predicate to indicate whether a value is valid. In the context
 of blockchain systems, for example, a value is not valid if it does not
 contain an appropriate hash of the last value (block) added to the blockchain.

%% file: consensus.tex
\section{Tendermint consensus algorithm} \label{sec:tendermint}

\newcommand\Disseminate{\textbf{Disseminate}}

\newcommand\Proposal{\mathsf{PROPOSAL}}
\newcommand\ProposalPart{\mathsf{PROPOSAL\mbox{-}PART}}
\newcommand\PrePrepare{\mathsf{INIT}} \newcommand\Prevote{\mathsf{PREVOTE}}
\newcommand\Precommit{\mathsf{PRECOMMIT}}
\newcommand\Decision{\mathsf{DECISION}}

\newcommand\ViewChange{\mathsf{VC}}
\newcommand\ViewChangeAck{\mathsf{VC\mbox{-}ACK}}
\newcommand\NewPrePrepare{\mathsf{VC\mbox{-}INIT}}
\newcommand\coord{\mathsf{proposer}}

\newcommand\newHeight{newHeight} \newcommand\newRound{newRound}
\newcommand\nil{nil} \newcommand\id{id} \newcommand{\propose}{propose}
\newcommand\prevote{prevote} \newcommand\prevoteWait{prevoteWait}
\newcommand\precommit{precommit} \newcommand\precommitWait{precommitWait}
\newcommand\commit{commit}

\newcommand\timeoutPropose{timeoutPropose}
\newcommand\timeoutPrevote{timeoutPrevote}
\newcommand\timeoutPrecommit{timeoutPrecommit}
\newcommand\proofOfLocking{proof\mbox{-}of\mbox{-}locking}

\begin{algorithm}[htb!] \def\baselinestretch{1} \scriptsize\raggedright
	\begin{algorithmic}[1] 
		\SHORTSPACE 
		\INIT{} 
		\STATE $h_p := 0$  
		\COMMENT{current height, or consensus instance we are currently executing} 
		\STATE $round_p := 0$   \COMMENT{current round number}
		\STATE $step_p  \in \set{\propose, \prevote, \precommit}$
		\STATE $decision_p[] := nil$ 
		\STATE $lockedValue_p := nil$ 
		\STATE $lockedRound_p := -1$ 
		\STATE $validValue_p := nil$ 
		\STATE $validRound_p := -1$ 
		\ENDINIT 
		\SHORTSPACE 
		\STATE \textbf{upon} start \textbf{do} $StartRound(0)$ 
		\SHORTSPACE 
		\FUNCTION{$StartRound(round)$} \label{line:tab:startRound} 
		\STATE $round_p \assign round$ 
		\STATE $step_p \assign \propose$ 
		\IF{$\coord(h_p, round_p) = p$}
		\IF{$validValue_p \neq \nil$} \label{line:tab:isThereLockedValue}
		\STATE $proposal \assign validValue_p$ \ELSE \STATE $proposal \assign
		getValue()$ 
		\label{line:tab:getValidValue} 
		\ENDIF 	  
		\STATE \Broadcast\ $\li{\Proposal,h_p, round_p, proposal, validRound_p}$
		\label{line:tab:send-proposal} 
		\ELSE 
		\STATE \textbf{schedule} $OnTimeoutPropose(h_p,
		round_p)$ to be executed \textbf{after} $\timeoutPropose(round_p)$ 
		\ENDIF
		\ENDFUNCTION
		
		\SPACE 
		\UPON{$\li{\Proposal,h_p,round_p, v, -1}$ \From\ $\coord(h_p,round_p)$
			\With\ $step_p = \propose$} \label{line:tab:recvProposal}
			\IF{$valid(v) \wedge (lockedRound_p = -1  \vee lockedValue_p = v$)}
			\label{line:tab:accept-proposal-2} 
				\STATE \Broadcast \ $\li{\Prevote,h_p,round_p,id(v)}$  
				\label{line:tab:prevote-proposal}	
			\ELSE
			\label{line:tab:acceptProposal1}		
				\STATE \Broadcast \ $\li{\Prevote,h_p,round_p,\nil}$  
				\label{line:tab:prevote-nil}	
			\ENDIF
				\STATE $step_p \assign \prevote$ \label{line:tab:setStateToPrevote1} 
		\ENDUPON
		
		\SPACE 
		\UPON{$\li{\Proposal,h_p,round_p, v, vr}$ \From\ $\coord(h_p,round_p)$
			\textbf{AND} $2f+1$ $\li{\Prevote,h_p, vr,id(v)}$  \With\ $step_p = \propose \wedge (vr \ge 0 \wedge vr < round_p)$}
		\label{line:tab:acceptProposal} 
		\IF{$valid(v) \wedge (lockedRound_p \le vr
			\vee lockedValue_p = v)$} \label{line:tab:cond-prevote-higher-proposal}	
			\STATE \Broadcast \ $\li{\Prevote,h_p,round_p,id(v)}$
			\label{line:tab:prevote-higher-proposal}		 
		\ELSE
			\label{line:tab:acceptProposal2}		
			\STATE \Broadcast \ $\li{\Prevote,h_p,round_p,\nil}$  
			\label{line:tab:prevote-nil2}	
		\ENDIF
		\STATE $step_p \assign \prevote$ \label{line:tab:setStateToPrevote3} 	 
		\ENDUPON
		
		\SPACE 
		\UPON{$2f+1$ $\li{\Prevote,h_p, round_p,*}$ \With\ $step_p = \prevote$ for the first time}
		\label{line:tab:recvAny2/3Prevote} 
		\STATE \textbf{schedule} $OnTimeoutPrevote(h_p, round_p)$ to be executed \textbf{after}  $\timeoutPrevote(round_p)$ \label{line:tab:timeoutPrevote} 
		\ENDUPON
		
		\SPACE 
		\UPON{$\li{\Proposal,h_p,round_p, v, *}$ \From\ $\coord(h_p,round_p)$
			\textbf{AND} $2f+1$ $\li{\Prevote,h_p, round_p,id(v)}$  \With\ $valid(v) \wedge step_p \ge \prevote$ for the first time}
		\label{line:tab:recvPrevote} 
		\IF{$step_p = \prevote$}	
			\STATE $lockedValue_p \assign v$                \label{line:tab:setLockedValue} 
			\STATE $lockedRound_p \assign round_p$   \label{line:tab:setLockedRound} 
			\STATE \Broadcast \ $\li{\Precommit,h_p,round_p,id(v))}$  
			\label{line:tab:precommit-v}	
			\STATE $step_p \assign \precommit$ \label{line:tab:setStateToCommit} 
		\ENDIF 
		\STATE $validValue_p \assign v$ \label{line:tab:setValidRound} 
		\STATE $validRound_p \assign round_p$ \label{line:tab:setValidValue} 
		\ENDUPON
		
		\SHORTSPACE 
		\UPON{$2f+1$ $\li{\Prevote,h_p,round_p, \nil}$ 
			\With\ $step_p = \prevote$} 
			\STATE \Broadcast \ $\li{\Precommit,h_p,round_p, \nil}$
			\label{line:tab:precommit-v-1} 
			\STATE $step_p \assign \precommit$ 
		\ENDUPON
		
		\SPACE 
		\UPON{$2f+1$ $\li{\Precommit,h_p,round_p,*}$ for the first time}
		\label{line:tab:startTimeoutPrecommit} 
			\STATE \textbf{schedule} $OnTimeoutPrecommit(h_p, round_p)$ to be executed \textbf{after} $\timeoutPrecommit(round_p)$ 
			 
		\ENDUPON 
		
		\SPACE 
		\UPON{$\li{\Proposal,h_p,r, v, *}$ \From\ $\coord(h_p,r)$ \textbf{AND}
			$2f+1$ $\li{\Precommit,h_p,r,id(v)}$ \With\ $decision_p[h_p] = \nil$}
		\label{line:tab:onDecideRule} 
			\IF{$valid(v)$} \label{line:tab:validDecisionValue} 
				\STATE $decision_p[h_p] = v$ \label{line:tab:decide} 
				\STATE$h_p \assign h_p + 1$ \label{line:tab:increaseHeight} 
				\STATE reset $lockedRound_p$, $lockedValue_p$, $validRound_p$ and $validValue_p$ to initial values 
				and empty message log 
				\STATE $StartRound(0)$   	
			\ENDIF 
		\ENDUPON
		
		\SHORTSPACE 
		\UPON{$f+1$ $\li{*,h_p,round, *, *}$ \textbf{with} $round > round_p$} 
		\label{line:tab:skipRounds} 
			\STATE $StartRound(round)$ \label{line:tab:nextRound2} 
		\ENDUPON
		
		\SHORTSPACE 
		\FUNCTION{$OnTimeoutPropose(height,round)$} \label{line:tab:onTimeoutPropose} 
		\IF{$height = h_p \wedge round = round_p \wedge step_p = \propose$} 
			\STATE \Broadcast \ $\li{\Prevote,h_p,round_p, \nil}$ 
		 	\label{line:tab:prevote-nil-on-timeout}	
		 	\STATE $step_p \assign \prevote$ 
		 \ENDIF	
		 \ENDFUNCTION
		
		\SHORTSPACE 
		\FUNCTION{$OnTimeoutPrevote(height,round)$} \label{line:tab:onTimeoutPrevote} 
		\IF{$height = h_p \wedge round = round_p \wedge step_p = \prevote$} 
			\STATE \Broadcast \ $\li{\Precommit,h_p,round_p,\nil}$   
			\label{line:tab:precommit-nil-onTimeout}
			\STATE $step_p \assign \precommit$ 
		\ENDIF	
		\ENDFUNCTION
		
		\SHORTSPACE 
		\FUNCTION{$OnTimeoutPrecommit(height,round)$} \label{line:tab:onTimeoutPrecommit} 
		\IF{$height = h_p \wedge round = round_p$}
			\STATE $StartRound(round_p + 1)$ \label{line:tab:nextRound} 
		\ENDIF
		\ENDFUNCTION	
	\end{algorithmic} \caption{Tendermint consensus algorithm}
	\label{alg:tendermint} 
\end{algorithm}

In this section we present the Tendermint Byzantine fault-tolerant consensus
algorithm. The algorithm is specified by the pseudo-code listing in
Algorithm~\ref{alg:tendermint}. We present the algorithm as a set of \emph{upon
rules} that are executed atomically\footnote{In case several rules are active
at the same time, the first rule to be executed is picked randomly. The
correctness of the algorithm does not depend on the order in which rules are
executed.}. We assume that processes exchange protocol messages using a gossip
protocol and that both received and sent messages are stored in a local message
log for every process. An upon rule is triggered once the message log contains
messages such that the corresponding condition evaluates to $\tt{true}$. The
condition that assumes reception of $X$ messages of a particular type and
content denotes reception of messages whose senders have aggregate voting power at
least equal to $X$. For example, the condition $2f+1$ $\li{\Precommit,h_p,r,id(v)}$,  
evaluates to true upon reception of $\Precommit$ messages for height $h_p$, 
a round $r$ and with value equal to $id(v)$ whose senders have aggregate voting 
power at least equal to $2f+1$. Some of the rules ends with "for the first time" constraint 
to denote that it is triggered only the first time a corresponding condition evaluates 
to $\tt{true}$. This is because those rules do not always change the state of algorithm 
variables so without this constraint, the algorithm could keep 
executing those rules forever. The variables with index $p$ are process local state
variables, while variables without index $p$ are value placeholders. The sign
$*$ denotes any value.    

We denote with $n$ the total voting power of processes in the system, and we
assume that the total voting power of faulty processes in the system is bounded
with a system parameter $f$.  The algorithm assumes that $n > 3f$, i.e., it
requires that the total voting power of faulty processes is smaller than one
third of the total voting power. For simplicity we present the algorithm for
the case $n = 3f + 1$.

The algorithm proceeds in rounds, where each round has a dedicated
\emph{proposer}. The mapping of rounds to proposers is known to all processes
and is given as a function $\coord(h, round)$, returning the proposer for
the round $round$ in the consensus instance $h$. We
assume that the proposer selection function is weighted round-robin, where
processes are rotated proportional to their voting power\footnote{A validator
with more voting power is selected more frequently, proportional to its power.
More precisely, during a sequence of rounds of size $n$, every process is
proposer in a number of rounds equal to its voting power.}. 
The internal protocol state transitions are triggered by message reception and 
by expiration of timeouts. There are three timeouts in Algorithm \ref{alg:tendermint}:
$\timeoutPropose$, $\timeoutPrevote$ and $\timeoutPrecommit$.
The timeouts prevent the algorithm from blocking and
waiting forever for some condition to be true, ensure that processes continuously 
transition between rounds, and guarantee that eventually (after GST) communication 
between correct processes is timely and reliable so they can decide. 
The last role is achieved by increasing the timeouts with every new round $r$, 
i.e, $timeoutX(r) = initTimeoutX + r*timeoutDelta$; 
they are reset for every new height (consensus
instance). 

Processes exchange the following messages in Tendermint: $\Proposal$,
$\Prevote$ and $\Precommit$. The $\Proposal$ message is used by the proposer of
the current round to suggest a potential decision value, while $\Prevote$ and
$\Precommit$ are votes for a proposed value. According to the classification of
consensus algorithms from \cite{RMS10:dsn}, Tendermint, like PBFT
\cite{CL02:tcs} and DLS \cite{DLS88:jacm}, belongs to class 3, so it requires
two voting steps (three communication exchanges in total) to decide a value.
The Tendermint consensus algorithm is designed for the blockchain context where
the value to decide is a block of transactions (ie. it is potentially quite
large, consisting of many transactions). Therefore, in the Algorithm
\ref{alg:tendermint} (similar as in \cite{CL02:tcs}) we are explicit about
sending a value (block of transactions) and a small, constant size value id (a
unique value identifier, normally a hash of the value, i.e., if $\id(v) =
\id(v')$, then $v=v'$). The $\Proposal$ message is the only one carrying the
value; $\Prevote$ and $\Precommit$ messages carry the value id.  A correct
process decides on a value $v$ in Tendermint upon receiving the $\Proposal$ for
$v$ and $2f+1$ voting-power equivalent $\Precommit$ messages for $\id(v)$ in
some round $r$. In order to send $\Precommit$ message for $v$ in a round $r$, a
correct process waits to receive the $\Proposal$ and $2f+1$ of the
corresponding $\Prevote$ messages in the round $r$. Otherwise, 
it sends $\Precommit$ message with a special $\nil$ value.  
This ensures that correct processes can $\Precommit$ only a 
single value (or $\nil$) in a round.  As
proposers may be faulty, the proposed value is treated by correct processes as
a suggestion (it is not blindly accepted), and a correct process tells others
if it accepted the $\Proposal$ for value $v$ by sending $\Prevote$ message for
$\id(v)$; otherwise it sends $\Prevote$ message with the special $\nil$ value. 

Every process maintains the following variables in the Algorithm
\ref{alg:tendermint}: $step$, $lockedValue$, $lockedRound$, $validValue$ and
$validRound$. The $step$ denotes the current state of the internal Tendermint
state machine, i.e., it reflects the stage of the algorithm execution in the
current round. The $lockedValue$ stores the most recent value (with respect to
a round number) for which a $\Precommit$ message has been sent. The
$lockedRound$ is the last round in which the process sent a $\Precommit$
message that is not $\nil$. We also say that a correct process locks a value
$v$ in a round $r$ by setting $lockedValue = v$ and $lockedRound = r$ before
sending $\Precommit$ message for $\id(v)$. As a correct process can decide a
value $v$ only if $2f+1$ $\Precommit$ messages for $\id(v)$ are received, this
implies that a possible decision value is a value that is locked by at least
$f+1$ voting power equivalent of correct processes. Therefore, any value $v$
for which $\Proposal$ and $2f+1$ of the corresponding $\Prevote$ messages are
received in some round $r$ is a \emph{possible decision} value. The role of the
$validValue$ variable is to store the most recent possible decision value; the
$validRound$ is the last round in which $validValue$ is updated. Apart from
those variables, a process also stores the current consensus instance ($h_p$,
called \emph{height} in Tendermint), and the current round number ($round_p$)
and attaches them to every message. Finally, a process also stores an array of
decisions, $decision_p$ (Tendermint assumes a sequence of consensus instances,
one for each height).

Every round starts by a proposer suggesting a value with the $\Proposal$
message (see line \ref{line:tab:send-proposal}). In the initial round of each
height, the proposer is free to chose the value to suggest. In the
Algorithm~\ref{alg:tendermint}, a correct process obtains a value to propose
using an external function    $getValue()$ that returns a valid value to
propose. In the following rounds, a correct proposer will suggest a new value
only if $validValue = \nil$; otherwise $validValue$ is proposed (see
lines~\ref{line:tab:isThereLockedValue}-\ref{line:tab:getValidValue}). 
In addition to the value proposed, the $\Proposal$ message also
contains the $validRound$ so other processes are informed about the last round
in which the proposer observed $validValue$ as a possible decision value.
Note that if a correct proposer $p$ sends $validValue$ with the $validRound$ in the
$\Proposal$, this implies that the process $p$ received $\Proposal$ and the
corresponding $2f+1$ $\Prevote$ messages for $validValue$ in the round
$validRound$. 
If a correct process sends $\Proposal$ message with $validValue$ ($validRound > -1$)
at time $t > GST$, by the \emph{Gossip communication} property, the
corresponding $\Proposal$ and the $\Prevote$ messages will be received by all
correct processes before time $t+\Delta$. Therefore, all correct processes will
be able to verify the correctness of the suggested value as it is supported by
the $\Proposal$ and the corresponding $2f+1$ voting power equivalent $\Prevote$
messages.   

A correct process $p$ accepts the proposal for a value $v$  (send $\Prevote$
for $id(v)$) if an external \emph{valid} function returns $true$ for the value
$v$, and if $p$ hasn't locked any value ($lockedRound = -1$) or $p$ has locked
the value $v$ ($lockedValue = v$); see the line
\ref{line:tab:accept-proposal-2}.  In case the proposed pair is $(v,vr \ge 0)$ and a
correct process $p$ has locked some value, it will accept
$v$ if it is a more recent possible decision value\footnote{As
explained above, the possible decision value in a round $r$ is the one for
which $\Proposal$ and the corresponding $2f+1$ $\Prevote$ messages are received
for the round $r$.}, $vr > lockedRound_p$,  or if $lockedValue = v$ 
(see line~\ref{line:tab:cond-prevote-higher-proposal}).  Otherwise, a correct
process will reject the proposal by sending $\Prevote$ message with $\nil$
value. A correct process will send $\Prevote$ message with $\nil$ value also in
case $\timeoutPropose$ expired (it is triggered when a correct process starts a
new round) and a process has not sent $\Prevote$ message in the current round
yet (see the line \ref{line:tab:onTimeoutPropose}). 

If a correct process receives $\Proposal$ message for some value $v$ and $2f+1$
$\Prevote$ messages for $\id(v)$, then it sends $\Precommit$ message with
$\id(v)$. Otherwise, it sends $\Precommit$ $\nil$. A correct process will send
$\Precommit$ message with $\nil$ value also in case $\timeoutPrevote$ expired
(it is started when a correct process sent $\Prevote$ message and received any
$2f+1$ $\Prevote$ messages)  and a process has not sent $\Precommit$ message in
the current round yet (see the line \ref{line:tab:onTimeoutPrecommit}).  A
correct process decides on some value $v$ if it receives in some round $r$
$\Proposal$ message for $v$ and $2f+1$ $\Precommit$ messages with $\id(v)$ (see
the line \ref{line:tab:decide}).  To prevent the algorithm from blocking and
waiting forever for this condition to be true, the Algorithm
\ref{alg:tendermint} relies on $\timeoutPrecommit$. It is triggered after a
process receives any set of $2f+1$ $\Precommit$ messages for the current round.
If the $\timeoutPrecommit$ expires and a process has not decided yet, the
process starts the next round (see the line \ref{line:tab:onTimeoutPrecommit}).
When a correct process $p$ decides, it starts the next consensus instance 
(for the next height). The \emph{Gossip communication} property ensures 
that $\Proposal$ and $2f+1$ $\Prevote$ messages that led $p$ to decide 
are eventually received by all correct processes, so they will also decide. 

\subsection{Termination mechanism}

Tendermint ensures termination by a novel mechanism that benefits from the
gossip based nature of communication (see \emph{Gossip communication}
property).  It requires managing two additional variables, $validValue$ and
$validRound$ that are then used by the proposer during the propose step as
explained above.   The $validValue$ and $validRound$ are updated to $v$ and $r$
by a correct process in a round $r$ when the process receives valid $\Proposal$
message for the value $v$ and the corresponding $2f+1$ $\Prevote$ messages for
$id(v)$ in the round $r$ (see the rule at line~\ref{line:tab:recvPrevote}).

We now give briefly the intuition how managing and proposing $validValue$
and $validRound$ ensures termination. Formal treatment is left for
Section~\ref{sec:proof}.  

The first thing to note is that during good period, because of the
\emph{Gossip communication} property, if a correct process $p$ locks a value
$v$ in some round $r$, all correct processes will update $validValue$ to $v$
and $validRound$ to $r$ before the end of the round $r$ (we prove this formally
in the Section~\ref{sec:proof}). The intuition is that messages that led to $p$
locking a value $v$ in the round $r$ will be gossiped to all correct processes
before the end of the round $r$, so it will update $validValue$ and
$validRound$ (the line~\ref{line:tab:recvPrevote}). Therefore, if a correct
process locks some value during good period, $validValue$ and $validRound$ are
updated by all correct processes so that the value proposed in the following
rounds will be acceptable by all correct processes. Note 
that it could happen that during good period, no correct process locks a value,
but some correct process $q$ updates $validValue$ and $validRound$ during some
round. As no correct process locks a value in this case, $validValue_q$ and
$validRound_q$ will also be acceptable by all correct processes as
$validRound_q > lockedRound_c$ for every correct process $c$ and as the
\emph{Gossip communication} property ensures that the corresponding $\Prevote$
messages that $q$ received in the round $validRound_q$ are received by all
correct processes $\Delta$ time later. 

Finally, it could happen that after GST, there is a long sequence of rounds in which 
no correct process neither locks a value nor update $validValue$ and $validRound$. 
In this case, during this sequence of rounds, the proposed value suggested by correct
processes was not accepted by all correct processes. Note that this sequence of rounds 
is always finite as at the beginning of every
round there is at least a single correct process $c$ such that $validValue_c$
and $validRound_c$ are acceptable by every correct process. This is true as
there exists a correct process $c$ such that for every other correct process
$p$, $validRound_c > lockedRound_p$ or $validValue_c = lockedValue_p$. This is
true as $c$ is the process that has locked a value in the most recent round
among all correct processes (or no correct process locked any value). Therefore,
eventually $c$ will be the proper in some round and the proposed value will be accepted
by all correct processes, terminating therefore this sequence of 
rounds. 

Therefore, updating $validValue$ and $validRound$ variables, and the
\emph{Gossip communication} property, together ensures that eventually, during
the good period, there exists a round with a correct proposer whose proposed
value will be accepted by all correct processes, and all correct processes will
terminate in that round. Note that this mechanism, contrary to the common
termination mechanism illustrated in the
Figure~\ref{ch3:fig:coordinator-change}, does not require exchanging any
additional information in addition to messages already sent as part of what is
normally being called "normal" case.

%% file: proof.tex
\section{Proof of Tendermint consensus algorithm} \label{sec:proof}

\begin{lemma} \label{lemma:majority-intersection} For all $f\geq 0$, any two
sets of processes with voting power at least equal to $2f+1$ have at least one
correct process in common.  \end{lemma}

\begin{proof} As the total voting power is equal to $n=3f+1$, we have $2(2f+1)
    = n+f+1$.  This means that the intersection of two sets with the voting
    power equal to $2f+1$ contains at least $f+1$ voting power in common, \ie,
    at least one correct process (as the total voting power of faulty processes
    is $f$). The result follows directly from this.  \end{proof}

\begin{lemma} \label{lemma:locked-decision_value-prevote-v} If $f+1$ correct
processes lock value $v$ in round $r_0$ ($lockedValue = v$ and $lockedRound =
r_0$), then in all rounds $r > r_0$, they send $\Prevote$ for $id(v)$ or
$\nil$.  \end{lemma}

\begin{proof} We prove the result by induction on $r$.

\emph{Base step $r = r_0 + 1:$} Let's denote with $C$ the set of correct
processes with voting power equal to $f+1$.  By the rules at
line~\ref{line:tab:recvProposal} and line~\ref{line:tab:acceptProposal}, the
processes from the set $C$ can't accept $\Proposal$ for any value different
from $v$ in round $r$, and therefore can't send a $\li{\Prevote,height_p,
r,id(v')}$ message, if $v' \neq v$. Therefore, the Lemma holds for the base
step.

\emph{Induction step from $r_1$ to $r_1+1$:} We assume that no process from the
set $C$ has sent $\Prevote$ for values different than $id(v)$ or $\nil$ until
round $r_1 + 1$. We now prove that the Lemma also holds for round $r_1 + 1$. As
processes from the set $C$ send $\Prevote$ for $id(v)$ or $\nil$ in rounds $r_0
\le r \le r_1$, by Lemma~\ref{lemma:majority-intersection} there is no value
$v' \neq v$ for which it is possible to receive $2f+1$ $\Prevote$ messages in
those rounds (i). Therefore, we have for all processes from the set $C$,
$lockedValue = v$ and $lockedRound \ge r_0$.   Let's assume by a contradiction
that a process $q$ from the set $C$ sends $\Prevote$ in round $r_1 + 1$ for
value $id(v')$, where $v' \neq v$. This is possible only by
line~\ref{line:tab:prevote-higher-proposal}.  Note that this implies that $q$
received $2f+1$ $\li{\Prevote,h_q, r,id(v')}$ messages, where $r > r_0$ and $r
< r_1 +1$ (see line~\ref{line:tab:cond-prevote-higher-proposal}). A
contradiction with (i) and Lemma~\ref{lemma:majority-intersection}.
\end{proof}	

\begin{lemma} \label{lemma:agreement} Algorithm~\ref{alg:tendermint} satisfies
Agreement.  \end{lemma}

\begin{proof} Let round $r_0$ be the first round of height $h$ such that some
    correct process $p$ decides $v$. We now prove that if some correct process
    $q$ decides $v'$ in some round $r \ge r_0$, then $v = v'$.

In case $r = r_0$, $q$ has received at least $2f+1$
$\li{\Precommit,h_p,r_0,id(v')}$  messages at line~\ref{line:tab:onDecideRule},
while $p$ has received at least $2f+1$ $\li{\Precommit,h_p,r_0,id(v)}$
messages.  By Lemma~\ref{lemma:majority-intersection} two sets of messages of
voting power $2f+1$ intersect in at least one correct process.  As a correct
process sends a single $\Precommit$ message in a round, then $v=v'$.

We prove the case $r > r_0$ by contradiction. By the
rule~\ref{line:tab:onDecideRule}, $p$ has received at least $2f+1$ voting-power
equivalent of $\li{\Precommit,h_p,r_0,id(v)}$ messages, i.e., at least $f+1$
voting-power equivalent correct processes have locked value $v$ in round $r_0$ and have
sent those messages (i). Let denote this set of messages with $C$.  On the
other side, $q$ has received at least $2f+1$ voting power equivalent of
$\li{\Precommit,h_q, r,id(v')}$ messages. As the voting power of all faulty
processes is at most $f$, some correct process $c$ has sent one of those
messages. By the rule at line~\ref{line:tab:recvPrevote}, $c$ has locked value
$v'$ in round $r$ before sending $\li{\Precommit,h_q, r,id(v')}$. Therefore $c$
has received $2f+1$ $\Prevote$ messages for $id(v')$ in round $r > r_0$ (see
line~\ref{line:tab:recvPrevote}). By Lemma~\ref{lemma:majority-intersection}, a
process from the set $C$ has sent $\Prevote$ message for $id(v')$ in round $r$.
A contradiction with (i) and Lemma~\ref{lemma:locked-decision_value-prevote-v}.
\end{proof}	

\begin{lemma} \label{lemma:agreement} Algorithm~\ref{alg:tendermint} satisfies
Validity.  \end{lemma}

\begin{proof} Trivially follows from the rule at line
\ref{line:tab:validDecisionValue} which ensures that only valid values can be
decided.  \end{proof}	

\begin{lemma} \label{lemma:round-synchronisation} If we assume that:
\begin{enumerate} 
    \item a correct process $p$ is the first correct process to
            enter a round $r>0$ at time $t > GST$ (for every correct process
            $c$, $round_c \le r$ at time $t$) 
    \item the proposer of round $r$ is
            a correct process $q$ 
    \item for every correct process $c$,
            $lockedRound_c \le validRound_q$ at time $t$ 
    \item $\timeoutPropose(r)
        > 2\Delta + \timeoutPrecommit(r-1)$, $\timeoutPrevote(r) > 2\Delta$ and
            $\timeoutPrecommit(r) > 2\Delta$, 
\end{enumerate} 
then all correct processes decide in round $r$ before $t + 4\Delta +
    \timeoutPrecommit(r-1)$.  
\end{lemma}	

\begin{proof} As $p$ is the first correct process to enter round $r$, it
    executed the line~\ref{line:tab:nextRound} after $\timeoutPrecommit(r-1)$
    expired. Therefore, $p$ received $2f+1$ $\Precommit$ messages in the round
    $r-1$ before time $t$. By the \emph{Gossip communication} property, all
    correct processes will receive those messages the latest at time $t +
    \Delta$. Correct processes that are in rounds $< r-1$ at time $t$ will
    enter round $r-1$ (see the rule at line~\ref{line:tab:nextRound2}) and
    trigger $\timeoutPrecommit(r-1)$ (see rule~\ref{line:tab:startTimeoutPrecommit})
    by time $t+\Delta$. Therefore, all correct processes will start round $r$
    by time $t+\Delta+\timeoutPrecommit(r-1)$ (i).
 
In the worst case, the process $q$ is the last correct process to enter round
$r$, so $q$ starts round $r$ and sends $\Proposal$ message for some value $v$
at time $t + \Delta + \timeoutPrecommit(r-1)$. Therefore, all correct processes
receive the $\Proposal$ message from $q$ the latest by time $t + 2\Delta +
\timeoutPrecommit(r-1)$. Therefore, if $\timeoutPropose(r) > 2\Delta +
\timeoutPrecommit(r-1)$, all correct processes will receive $\Proposal$ message
before $\timeoutPropose(r)$ expires. 

By (3) and the rules at line~\ref{line:tab:recvProposal} and
\ref{line:tab:acceptProposal}, all correct processes will accept the
$\Proposal$ message for value $v$ and will send a $\Prevote$ message for
$id(v)$ by time $t + 2\Delta + \timeoutPrecommit(r-1)$.  Note that by the
\emph{Gossip communication} property, the $\Prevote$ messages needed to trigger
the rule at line~\ref{line:tab:acceptProposal} are received before time $t +
\Delta$.  

By time $t + 3\Delta + \timeoutPrecommit(r-1)$, all correct processes will receive
$\Proposal$ for $v$ and $2f+1$ corresponding $\Prevote$ messages for $id(v)$.
By the rule at line~\ref{line:tab:recvPrevote}, all correct processes will send
a $\Precommit$ message (see line~\ref{line:tab:precommit-v}) for $id(v)$ by
time $t + 3\Delta + \timeoutPrecommit(r-1)$. Therefore, by time $t + 4\Delta +
\timeoutPrecommit(r-1)$, all correct processes will have received the $\Proposal$
for $v$ and $2f+1$ $\Precommit$ messages for $id(v)$, so they decide at
line~\ref{line:tab:decide} on $v$. 

This scenario holds if every correct process $q$ sends a $\Precommit$ message
before $\timeoutPrevote(r)$ expires, and if $\timeoutPrecommit(r)$ does not expire
before $t + 4\Delta + \timeoutPrecommit(r-1)$.  Let's assume that a correct process
$c_1$ is the first correct process to trigger $\timeoutPrevote(r)$ (see the rule
at line~\ref{line:tab:recvAny2/3Prevote}) at time $t_1 > t$. This implies that
before time $t_1$, $c_1$ received a $\Proposal$ ($step_{c_1}$ must be
$\prevote$ by the rule at line~\ref{line:tab:recvAny2/3Prevote}) and a set of
$2f+1$ $\Prevote$ messages.  By time $t_1 + \Delta$, all correct processes will
receive those messages. Note that even if some correct process was in the
smaller round before time $t_1$, at time $t_1 + \Delta$ it will start round $r$
after receiving those messages (see the rule at
line~\ref{line:tab:skipRounds}).  Therefore, all correct processes will send
their $\Prevote$ message for $id(v)$ by time $t_1 + \Delta$, and all correct
processes will receive those messages the by time $t_1 + 2\Delta$.  Therefore,
as $\timeoutPrevote(r) > 2\Delta$, this ensures that all correct processes receive
$\Prevote$ messages from all correct processes before their respective local
$\timeoutPrevote(r)$ expire.   

On the other hand, $\timeoutPrecommit(r)$ is triggered in a correct process $c_2$
after it receives any set of $2f+1$ $\Precommit$ messages for the first time.
Let's denote with $t_2 > t$ the earliest point in time $\timeoutPrecommit(r)$ is
triggered in some correct process $c_2$. This implies that $c_2$ has received
at least $f+1$ $\Precommit$ messages for $id(v)$ from correct processes, i.e.,
those processes have received $\Proposal$ for $v$ and $2f+1$ $\Prevote$
messages for $id(v)$ before time $t_2$. By the \emph{Gossip communication}
property, all correct processes will receive those messages by time $t_2 +
\Delta$, and will send $\Precommit$ messages for $id(v)$. Note that even if
some correct processes were at time $t_2$ in a round smaller than $r$, by the
rule at line~\ref{line:tab:skipRounds} they will enter round $r$ by time $t_2 +
\Delta$.  Therefore, by time $t_2 + 2\Delta$, all correct processes will
receive $\Proposal$ for $v$ and $2f+1$ $\Precommit$ messages for $id(v)$. So if
$\timeoutPrecommit(r) > 2\Delta$, all correct processes will decide before the
timeout expires.         \end{proof}

\begin{lemma} \label{lemma:validValue} If a correct process $p$ locks a value
    $v$ at time $t_0 > GST$ in some round $r$ ($lockedValue = v$ and
    $lockedRound = r$) and $\timeoutPrecommit(r) > 2\Delta$, then all correct
    processes set $validValue$ to $v$ and $validRound$ to $r$ before starting
    round $r+1$.  \end{lemma}
 
\begin{proof} In order to prove this Lemma, we need to prove that if the
    process $p$ locks a value $v$ at time $t_0$, then no correct process will
    leave round $r$ before time $t_0 + \Delta$ (unless it has already set
    $validValue$ to $v$ and $validRound$ to $r$). It is sufficient to prove
    this, since by the \emph{Gossip communication} property the messages that
    $p$ received at time $t_0$ and that triggered rule at
    line~\ref{line:tab:recvPrevote} will be received by time $t_0 + \Delta$ by
    all correct processes, so all correct processes that are still in round $r$
    will set $validValue$ to $v$ and $validRound$ to $r$ (by the rule at
    line~\ref{line:tab:recvPrevote}). To prove this, we need to compute the
    earliest point in time a correct process could leave round $r$ without
    updating $validValue$ to $v$ and $validRound$ to $r$ (we denote this time
    with $t_1$). The Lemma is correct if $t_0 + \Delta < t_1$. 

If the process $p$ locks a value $v$ at time $t_0$, this implies that $p$
received the valid $\Proposal$ message for $v$ and $2f+1$
$\li{\Prevote,h,r,id(v)}$ at time $t_0$. At least $f+1$ of those messages are
sent by correct processes. Let's denote this set of correct processes as $C$. By
Lemma~\ref{lemma:majority-intersection} any set of $2f+1$ $\Prevote$ messages
in round $r$ contains at least a single message from the set $C$. 

Let's denote as time $t$ the earliest point in time a correct process, $c_1$, triggered
$\timeoutPrevote(r)$. This implies that $c_1$ received $2f+1$ $\Prevote$ messages
(see the rule at line \ref{line:tab:recvAny2/3Prevote}), where at least one of
those messages was sent by a process $c_2$ from the set $C$.  Therefore, process
$c_2$ had received $\Proposal$ message before time $t$. By the \emph{Gossip
communication} property, all correct processes will receive $\Proposal$ and
$2f+1$ $\Prevote$ messages for round $r$ by time $t+\Delta$. The latest point
in time $p$ will trigger $\timeoutPrevote(r)$ is $t+\Delta$\footnote{Note that
even if $p$ was in smaller round at time $t$ it will start round $r$ by time
$t+\Delta$.}.  So the latest point in time $p$ can lock the value $v$ in
round $r$ is $t_0 = t+\Delta+\timeoutPrevote(r)$ (as at this point
$\timeoutPrevote(r)$ expires, so a process sends $\Precommit$ $\nil$ and updates
$step$ to $\precommit$, see line \ref{line:tab:onTimeoutPrevote}).  

Note that according to the Algorithm \ref{alg:tendermint}, a correct process
can not send a $\Precommit$ message before receiving $2f+1$ $\Prevote$
messages.  Therefore, no correct process can send a $\Precommit$ message in
round $r$ before time $t$. If a correct process sends a $\Precommit$ message
for $\nil$, it implies that it has waited for the full duration of
$\timeoutPrevote(r)$ (see line
\ref{line:tab:precommit-nil-onTimeout})\footnote{The other case in which a
correct process $\Precommit$ for $\nil$ is after receiving $2f+1$ $Prevote$ for
$\nil$ messages, see the line \ref{line:tab:precommit-v-1}. By
Lemma~\ref{lemma:majority-intersection}, this is not possible in round $r$.}.
Therefore, no correct process can send $\Precommit$ for $\nil$ before time $t +
\timeoutPrevote(r)$ (*).

A correct process $q$ that enters round $r+1$ must wait (i) $\timeoutPrecommit(r)$
(see line \ref{line:tab:nextRound}) or (ii) receiving $f+1$ messages from the
round $r+1$ (see the line \ref{line:tab:skipRounds}).  In the former case, $q$
receives $2f+1$ $\Precommit$ messages before starting $\timeoutPrecommit(r)$. If
at least a single $\Precommit$ message from a correct process (at least $f+1$
voting power equivalent of those messages is sent by correct processes) is for
$\nil$, then $q$ cannot start round $r+1$ before time $t_1 = t +
\timeoutPrevote(r) + \timeoutPrecommit(r)$ (see (*)). Therefore in this case we have:
$t_0 + \Delta < t_1$, i.e., $t+2\Delta+\timeoutPrevote(r) <  t + \timeoutPrevote(r) +
\timeoutPrecommit(r)$, and this is true whenever $\timeoutPrecommit(r) > 2\Delta$, so
Lemma holds in this case. 

If in the set of $2f+1$ $\Precommit$ messages $q$ receives, there is at least a
single $\Precommit$ for $id(v)$ message from a correct process $c$, then $q$
can start the round $r+1$ the earliest at time $t_1 = t+\timeoutPrecommit(r)$. In
this case, by the \emph{Gossip communication} property, all correct processes
will receive $\Proposal$ and $2f+1$ $\Prevote$ messages (that $c$ received
before time $t$) the latest at time $t+\Delta$. Therefore, $q$ will set
$validValue$ to $v$ and $validRound$ to $r$ the latest at time $t+\Delta$. As
$t+\Delta < t+\timeoutPrecommit(r)$, whenever $\timeoutPrecommit(r) > \Delta$, the
Lemma holds also in this case.    

In case (ii), $q$ received at least a single message from a correct process $c$
from the round $r+1$. The earliest point in time $c$ could have started round
$r+1$ is $t+\timeoutPrecommit(r)$ in case it received a $\Precommit$ message for
$v$ from some correct process in the set of $2f+1$ $\Precommit$ messages it
received. The same reasoning as above holds also in this case, so $q$ set
$validValue$ to $v$ and $validRound$ to $r$ the latest by time $t+\Delta$. As
$t+\Delta < t+\timeoutPrecommit(r)$, whenever $\timeoutPrecommit(r) > \Delta$, the
Lemma holds also in this case.    \end{proof}

\begin{lemma} \label{lemma:agreement} Algorithm~\ref{alg:tendermint} satisfies
Termination.  \end{lemma}

\begin{proof} Lemma~\ref{lemma:round-synchronisation} defines a scenario in
    which all correct processes decide. We now prove that within a bounded
    duration after GST such a scenario will unfold. Let's assume that at time
    $GST$ the highest round started by a correct process is $r_0$, and that
    there exists a correct process $p$ such that the following holds: for every
    correct process $c$, $lockedRound_c \le validRound_p$. Furthermore, we
    assume that $p$ will be the proposer in some round $r_1 > r$ (this is
    ensured by the $\coord$ function). 

We have two cases to consider. In the first case, for all rounds $r \ge r_0$
and $r < r_1$, no correct process locks a value (set $lockedRound$ to $r$). So
in round $r_1$ we have the scenario from the
Lemma~\ref{lemma:round-synchronisation}, so all correct processes decides in
round $r_1$.  

In the second case, a correct process locks a value $v$ in round $r_2$, where
$r_2 \ge r_0$ and $r_2 < r_1$.  Let's assume that $r_2$ is the highest round
before $r_1$ in which some correct process $q$ locks a value. By Lemma
\ref{lemma:validValue} at the end of round $r_2$ the following holds for all
correct processes $c$: $validValue_c = lockedValue_q$ and $validRound_c = r_2$.
Then in round $r_1$, the conditions for the
Lemma~\ref{lemma:round-synchronisation} holds, so all correct processes decide.
\end{proof}

%% file: conclusion.tex
\section{Conclusion} \label{sec:conclusion}

We have proposed a new Byzantine-fault tolerant consensus algorithm that is the
core of the Tendermint BFT SMR platform. The algorithm is designed for the wide
area network with high number of mutually distrusted nodes that communicate
over gossip based peer-to-peer network. It has only a single mode of execution
and the communication pattern is very similar to the "normal" case of the
state-of-the art PBFT algorithm. The algorithm ensures termination with a novel
mechanism that takes advantage of the gossip based communication between nodes.
The proposed algorithm and the proofs are simple and elegant, and we believe
that this makes it easier to understand and implement correctly.   

\section*{Acknowledgment}

We would like to thank Anton Kaliaev, Ismail Khoffi and Dahlia Malkhi for comments on an earlier version of the paper. We also want to thank Marko Vukolic, Ming Chuan Lin, Maria Potop-Butucaru, Sara Tucci, Antonella Del Pozzo and Yackolley Amoussou-Guenou for pointing out the liveness issues
in the previous version of the algorithm. Finally, we want to thank the Tendermint team members and all project contributors for making Tendermint such a great platform.